\pdfoutput=1
\documentclass[sigconf,screen,authorversion]{acmart}

\usepackage{algorithm}
\usepackage[noend]{algpseudocode}
\algrenewcomment[1]{\hfill \(\triangleright\) \emph{\small #1}} 
\algnewcommand{\InlineIf}[2]{ 
  \State \algorithmicif\ #1\ \algorithmicthen\ #2}
\algnewcommand{\InlineIfElse}[3]{ 
  \State \algorithmicif\ #1\ \algorithmicthen\ #2\ \algorithmicelse\ #3}
\algnewcommand{\InlineFor}[2]{\algorithmicfor\ #1\ \algorithmicdo\ #2} 
\usepackage{mathdots}
\usepackage[capitalise]{cleveref}

\newcommand{\bigO}[1]{O(#1)} 
\newcommand{\softO}[1]{\mathchoice{\tilde{O}\left(#1\right)}{O\tilde{~}(#1)}{O\tilde{~}(#1)}{O\tilde{~}(#1)}} 
\newcommand{\expmm}{\omega} 
\newcommand{\NN}{\mathbb{N}} 
\newcommand{\ZZp}{\mathbb{Z}_{> 0}} 
\newcommand{\field}{\mathbb{K}} 
\newcommand{\yAlg}{\alg[y]} 
\newcommand{\abRing}{\field[\alpha,\beta]} 
\newcommand{\xRing}{\field[x]} 
\newcommand{\prc}{d} 
\newcommand{\alg}{\mathbb{A}} 
\newcommand{\ann}[1]{\operatorname{Ann}(#1)} 
\newcommand{\genby}[1]{\langle #1 \rangle} 
\newcommand{\GENBY}[1]{\left\langle #1 \right\rangle} 
\newcommand{\mats}[2]{\field^{#1 \times #2}} 
\newcommand{\xmats}[2]{\xRing^{#1 \times #2}} 
\newcommand{\abmats}[2]{\abRing^{#1 \times #2}} 
\newcommand{\sdim}{n} 
\newcommand{\order}{\delta} 
\newcommand{\seqeSet}{\alg^{\sdim}} 
\newcommand{\seqSetExpand}{(\seqeSet)^\NN} 
\newcommand{\seqSet}{\mathcal{S}} 
\newcommand{\seq}{\boldsymbol{s}} 
\newcommand{\seqe}[1]{S_{#1}} 
\newcommand{\biseqeSet}{\field^{\sdim}} 
\newcommand{\biseqSetExpand}{(\biseqeSet)^{\NN^2}} 
\newcommand{\biseqSet}{\mathfrak{S}} 
\newcommand{\biseq}{\boldsymbol{\sigma}} 
\newcommand{\biseqe}[1]{\zeta_{#1}} 
\newcommand{\gens}{G_{\seq}} 
\newcommand{\tseqSetExpand}[1]{(\seqeSet)^{#1}} 
\newcommand{\tseqSet}[1]{\mathcal{S}_{#1}} 
\newcommand{\tseq}[1]{\boldsymbol{s}_{#1}} 
\newcommand{\tdim}{e} 
\newcommand{\hk}[1]{H_{#1}} 
\newcommand{\kerhk}[1]{\mathcal{K}_{#1}} 
\newcommand{\ord}{\preccurlyeq} 
\newcommand{\ordLex}{\ord_{\mathrm{lex}} } 
\newcommand{\modApp}[2]{\mathcal{A}_{#1}{(#2)}} 

\copyrightyear{2021}
\acmYear{2021}
\setcopyright{acmlicensed}\acmConference[ISSAC '21]{Proceedings of the 2021 International Symposium on Symbolic and Algebraic Computation}{July 18--23, 2021}{Virtual Event, Russian Federation}
\acmBooktitle{Proceedings of the 2021 International Symposium on Symbolic and Algebraic Computation (ISSAC '21), July 18--23, 2021, Virtual Event, Russian Federation}
\acmPrice{15.00}
\acmDOI{10.1145/3452143.3465533}
\acmISBN{978-1-4503-8382-0/21/07}

\begin{document}
\fancyhead{}

\title[Algorithms for Linearly Recurrent Sequences of Truncated Polynomials]{Algorithms for Linearly Recurrent Sequences \texorpdfstring{\\}{} of Truncated Polynomials}

\author{Seung Gyu Hyun}
\affiliation{%
  \institution{{\normalsize University of Waterloo}}
  \city{Waterloo, ON} 
  \state{Canada} 
}

\author{Vincent Neiger}
\affiliation{%
  \institution{{\normalsize Univ. Limoges, CNRS, XLIM, UMR\,7252}}
  \city{F-87000 Limoges} 
  \state{France} 
}

\author{\'Eric Schost}
\affiliation{%
  \institution{{\normalsize University of Waterloo}}
  \city{Waterloo, ON} 
  \state{Canada} 
}


\begin{abstract}
  Linear recurrent sequences are those whose elements are defined as linear
  combinations of preceding elements, and finding recurrence relations is a
  fundamental problem in computer algebra. In this paper, we focus on sequences
  whose elements are vectors over the ring $\alg=\xRing/\genby{x^\prc}$ of
  truncated polynomials. Finding the ideal of their recurrence relations has
  applications such as the computation of minimal polynomials and
  determinants of sparse matrices over $\alg$. We present three methods for
  finding this ideal: a Berlekamp-Massey-like approach due to Kurakin, one
  which computes the kernel of some block-Hankel matrix over \(\alg\) via a
  minimal approximant basis, and one based on bivariate Pad\'e approximation.
  We propose complexity improvements for the first two methods, respectively by
  avoiding the computation of redundant relations and by exploiting the Hankel
  structure to compress the approximation problem. Then we confirm these
  improvements empirically through a C++ implementation, and we discuss the
  above-mentioned applications. 
  \vspace{0.1cm}
\end{abstract}


\maketitle

\vspace{0.7cm} 
\section{Introduction}
\label{sec:intro}

Linear recurrences appear in many domains of computer science and mathematics,
and computing recurrence relations efficiently is a fundamental problem in
computer algebra. More specifically, given a sequence of elements in
\(\field^r\) for some field \(\field\) and integer \(r>0\), we seek a
representation of its \emph{annihilator}, which is a polynomial ideal
corresponding to all recurrence relations which are satisfied by the sequence;
the polynomials in the annihilator are said to \emph{cancel} the sequence. In
dimension \(r=1\), the Berlekamp-Massey algorithm \cite{berlekamp68,massey69}
computes the unique monic univariate polynomial of minimal degree that cancels
the sequence. Sakata extended this algorithm first to dimension \(2\)
\cite{Sakata1988} and then to the general case \(r>1\)
\cite{sakata1990extension}; see also Norton and Fitzpatrick's extension to
\(r>1\) \cite{FitzpatrickNorton90}. Recent work includes variants of Sakata's
algorithm such as one which handles relations that are satisfied by several
sequences simultaneously \cite{Sakata2009}, approaches relating the problem to
the kernel of a multi-Hankel matrix and exploiting either fast linear algebra
\cite{berthomieu17} or a process similar to Gram-Schmidt orthogonalization
\cite{Mourrain2017}, and an algorithm relying directly on multivariate
polynomial arithmetic \cite{BerthomieuFaugere2018}. As for the representation
of the output, all these algorithms compute a Gr\"obner basis or a border basis
of the annihilator.

In this paper, we focus on computing recurrence relations for sequences whose
elements are in $\alg^n$, where $\alg = \field[x]/\genby{x^\prc}$.
This problem can be solved using a specialization of Kurakin's algorithm
\cite{kurakin98,Kurakin2000FPSAC}, as detailed in \cref{sec:kurakin}, where we
explicitly describe the output generating set of the annihilator as a
lexicographic Gr\"obner basis of some bivariate ideal. We derive a cost bound
of $\softO{\order \prc (\sdim^2 \order \prc + \sdim^\expmm \prc)}$ operations
in \(\field\), where $\order$ is the order of recurrence (see
\cref{sec:preliminaries:sequences}), and \(\expmm\) is an exponent for matrix
multiplication over \(\field\) \cite{CopWin90,LeGall14,AlmanWilliams2021}.
Because the Gr\"obner bases computed by Kurakin's algorithm are often
non-minimal, in \cref{sec:modified_kurakin} we propose a modified algorithm
which aims at limiting as much as possible the computation of these extraneous
generators.  This lowers the cost to $\softO{\order \prc^* (\sdim^2 \order \prc
+ \sdim^\expmm \prc)}$, where $\prc^*$ is a number arising in the algorithm as
an upper bound on the cardinality $\prc_{\mathrm{opt}}$ of minimal Gr\"obner
bases of the annihilator. In \cref{sec:experiment}, we observe empirically that
\(\prc^*\) is often close or equal to $\prc_{\mathrm{opt}}$.

Despite the improvement, the above cost bound still has a dependence at least
quadratic in the dimension $\sdim$. Our interest in the case \(\sdim\gg 1\) is
motivated among others by the following fact: given a zero-dimensional ideal
$\mathcal{I} \in \field[x,y]$, one can recover a Gr\"obner basis of it via
$\mathcal{I} = \ann{\seq}$ for some well-chosen $\seq \in \alg^\NN$ only if
$\field[x,y]/\mathcal{I}$ has the \emph{Gorenstein} property
\cite{macaulay,grobner35}. When that is not the case, one can recover a basis
of $\mathcal{I}$ via the annihilator of \emph{several} sequences
simultaneously, which means precisely \(\sdim>1\). For large \(\sdim\), we
compute the annihilator via a minimal approximant basis of a block-Hankel
matrix over \(\alg\) constructed from $\seq$. Computing this approximant basis
via the algorithm \Call{PM-Basis}{} of \cite{GiJeVi03} leads to a complexity of
\(\softO{\order^\expmm \sdim \prc}\) operations in \(\field\)
(\cref{sec:pmbasis:basic}). We then propose a novel improvement of this minimal
approximant basis computation, based on a randomized compression of the input
matrix which leverages its block-Hankel structure, reducing the cost to
\(\softO{\order^2 \sdim \prc + \order^\expmm \prc}\) operations in \(\field\)
(\cref{sec:pmbasis:compress}).

The four above algorithms have been implemented in C++ using the libraries NTL
\cite{NTL} and PML \cite{PML}, using Lazard's structural theorem \cite{lazard}
for generating examples of sequences; see \cref{sec:experiment} for more
details. Our experiments on a prime field \(\field\) highlight a good match
between cost bounds and practical running times, confirming also the benefit
obtained from the improvements of both Kurakin's algorithm and the plain
approximant basis approach.

Furthermore, in \cref{sec:bivariate_pade} we propose an algorithm with cost
quasi-linear in the order \(\order\), whereas the above cost bounds are at
least quadratic. For \(\prc\in\bigO{\order}\), we compute the annihilator via
the bivariate Pad\'e approximation algorithm of \cite{NaldiNeiger2020}: this
uses \(\softO{\prc^{\expmm+1} \order}\) operations in \(\field\), at the price
of restricting to \(\sdim\in\bigO{1}\).

Finally, in \cref{sec:applications} we mention applications to the computation
of minimal polynomials and determinants of sparse matrices over \(\alg\).  To
design Wiedemann-like algorithms \cite{wiedemann} for such matrices $A \in
\alg^{\mu \times \mu}$, we need to compute annihilators from sequences of the
form $(u^T A^i v)_{i\ge0} \in \alg^\NN$ for some vectors \(u\) and \(v\);
several such sequences may be needed, leading to the case \(\sdim>1\).

Sakata's \(2\)-dimensional algorithm shares similarities with the case
\(\sdim=1\) of Kurakin's algorithm, and has the same complexity $\bigO{\delta^2
\prc^2}$ \cite[Thm.\,3]{Sakata1988}. Apart from this, to the best of our
knowledge previous work has \(\sdim=1\) and considers \(r\)-dimensional
sequences over \(\field\) for an arbitrary \(r\ge 2\)
\cite{berthomieu17,BerthomieuFaugere2018,Mourrain2017}. Complexity in this
\(r\)-variate context is often expressed using the degree \(D\) of the
considered zero-dimensional ideal; here, \(\order \le D \le \order \prc\) and a
minimal Gr\"obner basis or a border basis will have at most
\(\min(\order,\prc)+1\) elements. The \textsc{Scalar-FGLM} algorithm has cost
\(\softO{\prc_{\mathrm{opt}}\order^{\expmm} \prc}\)
\cite[Prop.\,16]{berthomieu17}. Both the Artinian border basis and
\textsc{Polynomial-Scalar-FGLM} algorithms
\cite{Mourrain2017,BerthomieuFaugere2018} cost \(\bigO{D^2\order\prc}\), which
is \(\bigO{\order^3\prc}\) in the most favourable case \(D=\order\), and
\(\bigO{\order^3\prc^3}\) when \(D \in \Theta(\delta \prc)\) (which will be the
case in our experiments, see \cref{sec:experiment}).
In all cases, a better complexity bound can be achieved by one of our
algorithms outlined above.

While this is not reflected in the cost estimates above, Kurakin's algorithm
and our modified version are still affected by the shape of the staircase of
the computed Gr\"obner basis, due to early termination of the iterations and
late additions; we leave a more refined complexity analysis with respect to $D$
as future work.

\section{Linearly Recurrent Sequences}
\label{sec:preliminaries}

In this section, we review key facts about linearly recurrent sequences and
algorithmic tools used throughout the paper.

\vspace{-0.2cm}
\subsection{Recurrent sequences over \texorpdfstring{$\xRing/\genby{x^\prc}$}{K[x]/(x**d)}}
\label{sec:preliminaries:sequences}

We consider the set \(\seqSet=\seqSetExpand\) of \emph{(vector) sequences} over
the ring $\alg = \xRing/\genby{x^\prc}$ for some \(\prc \in \ZZp\), that is,
sequences $\seq = (\seqe{0},\seqe{1},\ldots)$ with each \(\seqe{k}\) in
\(\seqeSet\). Such a sequence is said to be \emph{linearly recurrent} if there
exist $\gamma \in \NN$ and $p_0,\ldots,p_\gamma \in \alg$ with $p_\gamma$
invertible such that
\begin{equation}
  \label{eq:cancel}
  p_0 \seqe{k} + \cdots + p_{\gamma-1} \seqe{k+\gamma-1} + p_\gamma \seqe{k+\gamma} = 0
  \text{ for all } k \ge 0;
\end{equation}
the \emph{order} of $\seq$ is the smallest such $\gamma$, denoted by \(\order\)
hereafter. A polynomial $p_0 + \cdots + p_\gamma y^\gamma$ in $\yAlg$ is said
to \emph{cancel} $\seq$ if $p_0,\dots,p_\gamma$ satisfies \cref{eq:cancel}
(without requiring that \(p_\gamma\) be invertible). The set of canceling
polynomials forms an ideal $\ann{\seq}$ in $\yAlg$, called the
\emph{annihilator} of \(\seq\). Thus $\seq$ is linearly recurrent of order
$\order$ if and only if there is a monic polynomial of degree $\order$ in
\(\ann{\seq}\): such polynomials are called \emph{generating polynomials} of
$\seq$. Unlike for sequences over fields, here there may be canceling
polynomials of degree less than $\order$, which prevents uniqueness of
generating polynomials; and there are sequences which are not linearly
recurrent but still admit a nonzero canceling polynomial (i.e.~\(\ann{\seq}
\neq \{0\}\)).

\begin{example}
  Consider \(\alg = \xRing / \genby{x^2}\) and the sequence \(\seq =
  (1,1+x,1,1+x,1,1+x,\ldots)\) in \(\alg^\NN\). Note that \(x\seq =
  (x,x,x,x,\ldots)\). This sequence has order \(\order=2\), a generating
  polynomial is \(y^2-1\), and a canceling polynomial of degree less than \(2\)
  is \(x(y-1)\). One can verify that \(\ann{\seq} = \genby{y^2-1,x(y-1)}\); in
  particular \(y^2 + x(y-1) - 1\) is also a generating polynomial. For any
  sequence \(\seq\) in \(\field^\NN\) which is not linearly recurrent, the
  sequence \(x\seq\) in \(\alg^\NN\) is not linearly recurrent but is canceled
  by \(x\), i.e.~\(x \in \ann{x\seq} \setminus \{0\}\).
\end{example}

Like for sequences over fields, here canceling polynomials can be characterized as denominators of the (vector)
generating series of the sequence, defined as $\gens = \sum_{k\ge0} \seqe{k}
y^{-k-1}$ in \((\alg[[y^{-1}]])^\sdim\).  In what follows, the elements of
\(\yAlg^\sdim\) are called polynomials, and for \(g = (g_1,\ldots,g_\sdim) \in
\yAlg^\sdim\) we define \(\deg(g) = \max_{1 \le j \le \sdim} \deg(g_j)\).

\begin{lemma}
  \label{lem:series_characterization}
  Let $\seq \in \seqSet$, let $\gens$ be its generating series, and let $p\in
  \yAlg$. Then, \(p \in \ann{\seq}\) if and only if the series $p\gens \in
  (\alg[[y^{-1}]])^\sdim$ is a polynomial, in which case \(\deg(p\gens) <
  \deg(p)\).
\end{lemma}


In this paper, we want to compute a generating set for $\ann{\seq}$, for
a linearly recurrent \(\seq \in \seqSet\), but for algorithms we typically 
only have access to a finite number of terms of the sequence. 
Suppose we have access to the partial sequence \(\tseq{\tdim} =
(\seqe{0},\ldots,\seqe{\tdim-1})\) in \(\tseqSet{\tdim} =
\tseqSetExpand{\tdim}\), for some \(\tdim \in \ZZp\).  Similar to
\cref{eq:cancel}, a polynomial $p_0 + \cdots + p_\gamma y^\gamma$ of degree
\(\gamma<\tdim\) cancels \(\tseq{\tdim}\) if 
\begin{align}
	\label{eq:partial}
	p_0 \seqe{k} + \cdots + p_\gamma \seqe{k+ \gamma} = 0
  \text{ for all } 0 \le k< \tdim-\gamma.
\end{align}
Like for sequences over fields, here polynomials of degree \(\gamma\) which cancel
\(\tseq{\tdim}\) also cancel the whole sequence \(\seq\), provided the
discrepancy between \(\tdim\) and \(\gamma\) is sufficiently large (namely,
\(\tdim\ge\gamma+\order\)).
\begin{lemma}
  \label{lem:characterization_partial}
  Let $\seq \in \seqSet$ be linearly recurrent of order \(\order\). For any
  \(\tdim\in\ZZp\) and any \(p\in \yAlg\) with \(\deg(p) \le \tdim - \order\),
  one has \(p \in \ann{\seq}\) if and only if \(p\) cancels \(\tseq{\tdim}\).
\end{lemma}

\vspace{-0.3cm}
\subsection{Bivariate interpretation and generating sets}
\label{sec:preliminaries:bivariate}

Uni-dimensional sequences of vectors in \(\seqeSet\) as above can be
interpreted as two-dimensional sequences of vectors in \(\biseqeSet\), that is,
sequences \(\biseq = (\biseqe{i,j})_{i,j\ge0}\) in \(\biseqSet =
\biseqSetExpand\). This is based on the natural injection \(\varphi :
\yAlg \to \abRing\) with
\((\varphi(x),\varphi(y))=(\alpha,\beta)\). 

Here we recall from \cite{Sakata1988,FitzpatrickNorton90} that a polynomial \(q
= \sum_{i,j} q_{ij} \alpha^i\beta^j\) in \(\abRing\) is said to
cancel a sequence \(\biseq = (\biseqe{i,j})_{i,j\ge0} \in \biseqSet\) if
\[
  \textstyle\sum_{i,j} q_{ij} \biseqe{i+k_1, j+k_2} = 0
  \text{ for all } k_1, k_2 \ge 0. 
\]
Then, let $\seq = (\seqe{0},\seqe{1},\ldots) \in \seqSet$, and define \(\biseq =
(\biseqe{i,j})_{i,j\ge0} \in \biseqSet\) such that
\(\biseqe{i,j}\in\biseqeSet\) is the coefficient of degree \(d-1-i\) of the
truncated polynomial vector \(\seqe{j} \in \seqeSet\) if \(i < d\), and
$\biseqe{i,j} = 0$ otherwise. Then, a polynomial $p\in \yAlg$ cancels $\seq$ if
and only if the polynomial $\varphi(p)$ cancels $\biseq$. Furthermore, the set
of polynomials in \(\abRing\) which cancel \(\biseq\) is an ideal of
\(\abRing\) which contains \(\alpha^\prc\), and this ideal is zero-dimensional
if and only if \(\seq\) is linearly recurrent.

In what follows, we define \(\bar\varphi(\mathcal{I}) = \genby{\{ \varphi(p)
\mid p\in \mathcal{I} \} \cup \{ \alpha^\prc \}}\) for any ideal
\(\mathcal{I}\) of \(\yAlg\), providing a correspondence between the ideals of
\(\yAlg\) and those of \(\abRing\) containing \(\alpha^\prc\).
For insight into possible ``nice'' generating sets for \(\ann{\seq}\), we
consider the lexicographic order \(\ordLex\) with \(\alpha\ordLex\beta\), and
use the fact that Gr\"obner bases of the ideals in \(\abRing\) for
this order are well understood \cite{lazard}. Below, unless mentioned
otherwise, we use \(\ordLex\) when some term order is needed, e.g.~leading
terms and Gr\"obner bases.

Consider a zero-dimensional ideal $\mathcal{I}$ in $\abRing$ that
contains a power of $\alpha$ and let $\mathcal{G}$ be its reduced Gr\"obner
basis. Let
\vspace{-0.1cm}
\[
  (\beta^{e_0}, \alpha^{d_1} \beta^{e_1}, \ldots, \alpha^{d_{t-1}} \beta^{e_{t-1}} , \alpha^{d_t})
\]
be the leading terms of the elements of $\mathcal{G}$ listed in decreasing
order, i.e.~the $e_i$'s are decreasing and the $d_i$'s are increasing. We set
$d_0=e_t=0$, and for $1 \le i \le t$ we set $\delta_i = d_i - d_{i-1}$, so that
$d_i = \delta_1 + \cdots + \delta_i$. Similarly, for $0 \le i < t$ we set
$\varepsilon_i = e_i-e_{i+1}$. Then write $\mathcal{G} = \{g_0, \ldots, g_t\}$,
with $g_i$ having leading term $\alpha^{d_{i}} \beta^{e_{i}} $; in particular
$g_t = \alpha^{d_t}=\alpha^{\delta_1 + \cdots + \delta_t}$ and $g_0$ is monic
in $\beta$.

Lazard's Theorem states the following \cite{lazard}: for $0 \le i \le t$ one
can write $g_i = \alpha^{d_i} \hat{g}_i$, with $\hat{g}_i$ monic of degree
$e_i$ in $\beta$. In addition, for $0\le i < t$, $\hat{g}_{i} =
g_{i}/\alpha^{d_{i}}$ is in the ideal generated by
\vspace{-0.1cm}
\[
  \genby{
    \hat{g}_{i+1},
    \alpha^{\delta_{i+2}} \hat{g}_{i+2},
    \dots,
    \alpha^{\delta_{i+2}+ \cdots + \delta_t}
  }
  =
  \GENBY{
    \frac{g_{i+1}}{\alpha^{d_{i+1}}},
    \frac{g_{i+2}}{\alpha^{d_{i+1}}},
    \dots,
    \frac{g_t}{\alpha^{d_{i+1}}}
  };
\]
in particular, $\alpha^{\delta_1}$ divides $g_1,\dots,g_{t}$. Lazard also
proved that a set of polynomials which satisfies these conditions is
necessarily a minimal Gr\"obner basis.

With the above notation, a minimal Gr\"obner basis of \(\mathcal{I}\) has
cardinality \(t+1\), with \(t \le \min(e_0,d_t)\) since
\(0=d_0<d_1<\cdots<d_t\) and \(0=e_t<\cdots<e_1<e_0\). Since for the reduced
Gr\"obner basis \(\mathcal{G}\) each polynomial \(g_i\) is represented by at
most \(e_0d_t\) coefficients in \(\field\), the total size of \(\mathcal{G}\)
in terms of field elements is at most \(e_0d_t\min(e_0,d_t)\). Finer bounds
for the cardinality and size of \(\mathcal{G}\) could be given using the vector
space dimension \(\dim_\field(\abRing/\mathcal{I})\).

\subsection{Univariate and bivariate approximation}
\label{sec:preliminaries:tools}

For a univariate polynomial matrix \(F \in \xmats{\mu}{\nu}\) and a positive
integer \(d\), we consider a free \(\xRing\)-module of rank \(\mu\) defined as
\vspace{-0.1cm}
\[
  \modApp{d}{F} = \{p \in \xmats{1}{\mu} \mid pF = 0 \bmod x^d\};
\]
its elements are called \emph{approximants for \(F\) at order \(d\)}
\cite{BarBul92,BecLab94}. Bases of such submodules can be represented as
\(\mu\times \mu\) nonsingular matrices over \(\xRing\) and are usually computed
in so-called \emph{reduced} forms \cite{Wolovich74} or the corresponding
canonical \emph{Popov} forms \cite{Popov72}. Extensions of these forms have
been defined to accommodate degree weights or degree constraints, and are
called \emph{shifted} reduced or Popov forms \cite{BarBul92,BecLab94,BeLaVi99}.
The algorithm \Call{PM-Basis}{} \cite{GiJeVi03} computes an approximant basis
in shifted reduced form in time \(\softO{\mu^{\expmm-1}(\mu+\nu) d}\); using
essentially two calls to this algorithm, one recovers the unique approximant
basis in shifted Popov form within the same cost bound \cite{JeaNeiVil2020}.

More generally, in the bivariate case with \(F \in \abmats{\mu}{\nu}\) and
\((d,e)\in\ZZp\), the set
\vspace{-0.1cm}
\[
  \modApp{d,e}{F} = \{p \in \abmats{1}{\mu} \mid pF = 0 \bmod (\alpha^d,\beta^e)\}
\]
is a \(\abRing\)-submodule of \(\abmats{1}{\mu}\) whose elements are called
\emph{approximants for \(F\) at order \((d,e)\)}. Such submodules are usually
represented by a \(\ord\)-Gr\"obner basis for some term order \(\ord\) on
\(\abmats{1}{\mu}\); for definitions of term orders and Gr\"obner bases for
submodules we refer to \cite{CoLiOSh05}. For \(\nu\le\mu\) algorithms based on
an iterative approach or on efficient linear algebra yield cost bounds in
\(\softO{\mu (\nu de)^2 + (\nu de)^3}\) and \(\softO{\mu(\nu de)^{\expmm-1} +
(\nu de)^\expmm}\) operations in \(\field\) respectively
\cite{Fitzpatrick97,NeigerSchost2020}, whereas a recent divide and conquer
approach costs \(\softO{(M^{\expmm}+M^2\nu)de}\), where \(M=\mu\min(d,e)\)
\cite[Prop.\,5.5]{NaldiNeiger2020}; in these cases the output is a minimal
Gr\"obner basis.

\section{Kurakin's algorithm}
\label{sec:kurakin}

In \cite{kurakin98}, Kurakin gives an algorithm based on the Berlekamp-Massey algorithm
that computes the annihilators of a partial sequence over a ring $R$ (and modules over $R$) that can be decomposed as
a disjoint union
\(
  R = \{0\} \cup R_0 \cup \cdots \cup R_{\prc-1}
\)
where
\[
  R_i = \{ r_i r^* \mid r^* \in R \text{ invertible} \}
  \text{ for some } r_i \in R.
\]
In this paper we consider $R = \alg =
\xRing/\genby{x^\prc}$; in this case the canonical choice is $r_i = x^i$, with
\[R_i = \{ x^i p^* \mid p^* \in \alg \text{ with nonzero constant term} \}.\]

Consider a partial sequence $\seq_e \in \tseqSet{e}$ of a
linearly recurrent \(\seq \in \seqSet\) of order \(\order\). Kurakin's algorithm computes $d$ polynomials $P_i \in \alg[y]$, $i = 0,\dots,d-1$, 
such that $P_i$ is a canceling polynomial of $\seq_e$ that has leading coefficient $x^i$ and is minimal in degree among all canceling polynomials with
leading coefficient $x^i$. Furthermore, one has $\ann{\seq} = \langle P_0,\dots, P_{d-1} \rangle$ provided $e \ge 2 \delta$ \cite[Thm.\,1]{Kurakin2000FPSAC}.

We first define three operations on sequences. Given a partial sequence $\seq_e$ and $c\in\alg$, 
$c \cdot \seq_e$ denotes multiplying $c$ to every element in $\seq_e$, while $y^j \cdot \seq_e$ denotes a shift of $j$ elements ---
that is, removing the first $j$ elements. Given another partial sequence $\hat\seq_{\hat{e}}$, the sum $\seq_{e} + \hat\seq_{\hat{e}}$ returns
the first $\min(e,\hat{e})$ elements of the two sequences added together element-wise.

Kurakin's algorithm iterates on $s = 0,\ldots, e-1$, keeping track of polynomials $P_{i,s}$ as well as partial sequences
$\seq_{e,i,s} = P_{i,s} \cdot \seq_{e} = \sum_{j=0}^{e-s} P_{i,s}[j]\cdot y^j\cdot \seq_{e} $,
where $P_{i,s}[j]$ is the $j$-th coefficient of $P_{i,s}$.
An invariant is that the leading coefficient of $P_{i,s}$ is $x^i$ for all $s$.
For each $s = 0, \ldots, e-1$, the algorithm essentially attempts to either create a zero by using the partial sequences 
from previous iterations with equal number of leading zeros (similar to Gaussian elimination), or shift the sequence if we cannot cancel this element. 

At each iteration $s$, let
$\mathcal{I}[k]$ be the \(\alg\)-submodule of \(\alg^\sdim\) generated by the elements $\seq_{e,i,s'}[k]$ for all $i = 0,\dots,d-1$ and $s'< s$ such that $\seq_{e,i,s'}$ has $k$ leading
zeros. Furthermore, let $\mathcal{P}[k,j]$ and $\mathcal{S}[k,j]$ be the corresponding polynomial and partial sequence to the $j$-th element in the basis of
$\mathcal{I}[k]$, $\mathcal{I}[k,j]$.
 At iteration $s$, if $\seq_{e,i,s}$ has $k$ leading zeros and $\seq_{e,i,s}[k] \in \mathcal{I}[k]$, then we can
find coefficients such that $\seq_{e,i,s}[k] - \sum_j c_j \mathcal{I}[k,j] = 0$ and  $\seq_{e,i,s} - \sum_j c_j \mathcal{S}[k,j]$ results in a
sequence with at least $k+1$ zeros since both sequences had $k$ leading zeros and we canceled $\seq_{e,i,s}[k]$. The algorithm terminates when all $\seq_{e,i,s} = 0$
(see Algorithm \ref{algo:kurakin}).

\begin{algorithm}[ht]
  \caption{\textsc{Kurakin}\( (\seq_e) \)}
  \label{algo:kurakin}
  \begin{algorithmic}[1]
    \Require{partial sequence $\seq_e$}
    \Ensure{minimal canceling polynomials of $\seq_e$}
    
    \For{$i = 0,\dots, d-1$}
    	\State set $P_{i,0} = x^i$ and $\seq_{e,i,0} = x^i \seq_{e}$
    	\State set $k$ to be index of first non-zero element of $\seq_{e,i,0}$
	\If{$\seq_{e,i,0}[k] \ne 0$}\label{line:check_zero1}
		\State add $\seq_{e,i,0}[k], P_{i,0}, \seq_{e,i,0}$ to $\mathcal{I}[k], \mathcal{P}[k], \mathcal{S}[k]$ resp.
	\EndIf
    \EndFor
    
    \For{$s = 1, \ldots, e-1$}
    	\For{$i = 0,\ldots d-1$}
    		\State set $t = 0$; $P_{i,s}^{(t)} = yP_{i,s-1}$;
        and shift $\seq_{e,i,s}^{(t)} = y\cdot \seq_{e,i,s-1}$
        \InlineIf {$\seq_{e,i,s}^{(t)} = 0$}{continue to next $i$} \label{line:goto}
        \State set $k$ to be the first non-zero index of $\seq_{e,i,s}^{(t)}$ \label{line:check_zero2}
        \InlineIf{$\seq_{e,i,s}^{(t)}[k] \notin \mathcal{I}[k]$}{continue to next $i$} \label{line:check_member}
        \State solve for $c_j$'s such that $\seq_{e,s,i}^{(t)}[k] - \sum_j c_j \mathcal{I}[k,j] = 0$ \label{line:soln}
        \State set $\seq_{e,i,s}^{(t+1)} = \seq_{e,i,s}^{(t)} - \sum_j c_j \mathcal{S}[k,j]$ \label{line:next_seq}
        \State set $P_{i,s}^{(t+1)} = P_{i,s}^{(t)} - \sum_j c_j \mathcal{P}[k,j]$ \label{line:next_pol}
        \State go to line \ref{line:goto} with $t = t+1$
    	\EndFor
	 \For{$i = 0,\ldots, d-1$}
	    	\State set $s_{e,i,s} = s_{e,i,s}^{(t)}$ and $P_{i,s} = P_{i,s}^{(t)}$
		\State set $k$ to be the index of first non-zero element of $s_{e,i,s}$
		\If {$s_{e,i,s}[k] \notin \mathcal{I}[k]$} \label{line:update_ideal}
			\State add $\seq_{e,i,s}[k], P_{i,s}, \seq_{e,i,s}$ to $\mathcal{I}[k], \mathcal{P}[k], \mathcal{S}[k]$ resp.
			\State reduce the basis of $\mathcal{I}[k]$ if needed \label{line:reduce}
		\EndIf
	    \EndFor
    \EndFor
    \For{$i = 0,\ldots,d-1$}
    	\State return $P_{i,s}$ that makes $\seq_{e,i,s} = 0$ for the first time
    \EndFor
  \end{algorithmic}
\end{algorithm}

We track the subiterations by the index $t$ for analysis; this does not play a role in the algorithm. 
Kurakin shows that the total number of subiterations across
all $s$ is $\bigO{e}$ per polynomial, bringing the total to $\bigO{ed}$ (\cite[Thm.\,2]{kurakin98}). 
However, the analysis of the runtime in \cite{kurakin98} treats all ring operations 
(including computing solution to line \ref{line:soln} of Algorithm
\ref{algo:kurakin}) as constant time operations, which is unrealistic over $\alg^n$. 
Thus, we will give a cost analysis in terms of number of field operations over $\field$.

We note that, since \(\alg^n\) is a free \(\xRing\)-module of rank \(n\) (with
a basis given by the canonical vectors of length \(n\)) and \(\xRing\) is a
principal ideal domain, any of its \(\xRing\)-submodule is free of rank at most
\(n\). As a consequence, the number of generators of $\mathcal{I}[k]$ is at
most $n$. This will allow us to bound the cost for solving submodule membership
as well as the equation $\seq_{e,s,i}^{(t)}[k] - \sum_j c_j \mathcal{I}[k,j] =
0$.

We can check membership $s_{e,i,s}[k] \in \mathcal{I}[k]$ and solve $s_{e,s,i}[k] - \sum c_j \mathcal{I}[k,j] = 0$ by finding the
right approximant basis of $$F = \begin{bmatrix} \mathcal{I}[k,0] & \cdots & \mathcal{I}[k,n-1] & s_{e,s,i}[k] \end{bmatrix}$$
in Popov form. Since $F$ has $n$ rows and at most $n+1$ columns, we can compute this in cost $\softO{n^\omega d}$ \cite{JeaNeiVil2020}. The reduction in line \ref{line:reduce} can be computed by the same approximant basis:
if $F$ has $n+1$ columns, there is a column in the approximant basis such that at least one entry has a nonzero constant term. 
By removing the corresponding $\mathcal{I}[k,j]$, we get a basis of $\mathcal{I}[k]$
of size $n$.

At lines \ref{line:next_seq} and
\ref{line:next_pol}, $S[k,j]$ and $P[k,j]$ have length and degree at most $e$ resp., making the cost of these lines 
$\softO{n(ned)} = \softO{n^2 e d}$.
Finally, using the fact that the total number of subiterations is bounded by $\bigO{ed}$, we arrive at the
total cost $\softO{ed (n^2 e d + n^\expmm d)}$.

We conclude by showing that the output of Algorithm \ref{algo:kurakin} is indeed a basis of $\ann{\seq}$ and that it forms a lexicographical Gr\"obner basis.

\begin{theorem} \label{theorem:kurakin}
  For each \(i \in \{0,\ldots,\prc-1\}\), let $P_i$ be a canceling polynomial
  of $\seq$ with leading coefficient $x^i$ that is minimal in degree among all
  polynomials with leading coefficient $x^i$. Then one has $\ann{\seq} =
  \langle P_0,\dots, P_{d-1} \rangle$. Furthermore,
  $\{\varphi(P_0),\cdots,\varphi(P_{d-1}), \alpha^d\}$ forms a Gr\"obner basis
  of $\bar\varphi(\ann{\seq})$ with respect to the lexicographic term order
  with $\alpha \ordLex \beta$.
\end{theorem}

\begin{proof}
	Suppose that there exists some $Q \in \alg[y]$ with leading coefficient $x^t$ that is in $\ann{\seq}$ but $Q \notin \langle P_0, \dots, P_{d-1} \rangle$.
	Note that for any polynomial in $\alg[y]$, we can always make the leading coefficient to be some $x^t$ by pulling out the minimal power of $x$ from
	the leading coefficient and multiplying by its inverse. Now, since we assumed minimality of degrees for $P_i$'s, $\deg(Q) > \deg(P_t)$ and
	$Q' = Q - y^{\deg{Q}-\deg{P_t}} P_t \in \ann{\seq}$ has degree less than $Q$.
	By normalizing the leading coefficient of $Q'$ to be some $x^{t'}$, we can repeat the
	same process and keep decreasing the degree. This process must terminate when we encounter some $Q'$ with leading coefficient $x^{t'}$ such
	that $\deg{Q'} < \deg{P_{t'}}$, or $Q' = 0$. 
	Both cases lead to contradictions; thus, such $Q$ cannot exist and $\ann{\seq} =  \langle P_0, \dots, P_{d-1} \rangle$.
	
	Next, let $\mathcal{G} = \{g_0, \dots, g_k \}$, $g_i \in \abRing$ with leading coefficient $x^{d_i}$, 
	be the minimal reduced (lexicographic) Gr\"obner basis of $\bar\varphi(\ann{\seq})$.
	We can turn $\mathcal{G}$ into another non-minimal Gr\"obner basis by adding the polynomials $a^c g_i$, for $c = 1,\dots, d_{i+1}-1$; we define
	the resulting basis as $\mathcal{G}' = \{ g_0', \cdots, g_d' \}$, with $g'_d = \alpha^d$ and each $g_i'$ has leading term $\alpha^i \beta^{r_i}$. Furthermore, define
	$u_i$ as the degree of $P_i$ such that $\varphi(P_i)$ has leading term $\alpha^i \beta^{u_i}$.
	
	For $i = 0, \dots, d$, we have that $u_i \ge r_i$, otherwise $\mathcal{G}'$ would not reduce $\varphi(P_i)$ to zero, which $\mathcal{G}'$ must since 
	$\varphi(P_i) \in \bar\varphi(\ann{\seq})$. We also have that $u_i \le r_i$ due to the assumed minimality of degree for $P_i$'s. Thus, the leading terms of
	$\{ \varphi(P_0),\dots, \varphi(P_{d-1}), \alpha^d \}$ generate the leading terms of $\bar\varphi(\ann{\seq})$.
\end{proof}

\section{Lazy algorithm based on Kurakin's}
\label{sec:modified_kurakin}

Kurakin's algorithm requires that we keep track of all $d$ possible generators, regardless of the actual number of generators needed. For example,
consider $\seq = (1,1,2,3,5,\dots) \in \alg^\NN$ with $\ann{\seq} = \langle y^2 - y - 1 \rangle$: Kurakin's algorithm returns
$\{ x^i (y^2-y-1), 0 \le i < d\}$. In this section, we outline a modified version of Kurakin's algorithm that attempts to avoid as
many extraneous computations as possible.

In the previous example, we can see that the polynomials associated with $x^i$, $i\ge 1$, were not useful. The next definition aims to qualify precisely
the usefulness of the monomial $x^i$.
\begin{definition}
	Let $P_{i,s}$ and $\seq_{e,i,s}$ be the polynomial and sequence at the end of step $s$ associated with monomial $x^i$. 
	A monomial $x^{i_2}$ is \emph{useful} wrt to $x^{i_1}$, $i_1<i_2$, at step $s$ if at least one of two conditions is true at the end of $s$:
	\begin{itemize}
		\item[U1.] $P_{i_2,s} \ne x^{i_2-i_1} P_{i_1,s}$
		\item[U2.] let $k_{i_1}$ and $k_{i_2}$ be the index of the first non-zero element of $\seq_{e,i_1,s}$ and $\seq_{e,i_2,s}$ resp., then
		$k_{i_1} \ne k_{i_2}$
	\end{itemize}
\end{definition}

Suppose a monomial $x^{i_2}$ is not useful wrt $x^{i_1}$ at step $s$, then by negating condition U1, we have $P_{i_2,s} = x^{i_2-i_1} P_{i_1,s}$. 
Due to negation of U2, $\seq_{e,i_2,s}$ is the zero sequence if and only if $\seq_{e,i_1,s}$ is the zero sequence; so either we return 
$P_{i_2,s} = x^{t_2-t_1} P_{i_1,s}$ or we do not terminate at this step for both monomials. Finally, since $k_{i_1} = k_{i_2}$ and 
$\seq_{e,i_2,s} = x^{i_2-i_1} \seq_{e,i_1,s}$, we always have that
$\seq_{e,i_2,s}[k_{i_2}] = x^{i_2-i_1} \seq_{e,i_1,s}[k_{i_1}] \in \left( \langle \seq_{e,i_1,s}[k_{i_1}] \rangle \cup \mathcal{I}[k_{i_1}] \right)$, 
meaning we can safely ignore
$\seq_{e,i_2,s}[k_{i_2}]$ when updating $\mathcal{I}[k_{i_2}]$ at the end of step $s$.
Thus, the negation of usefulness conditions U1 and U2 implies that any computation associated with $x^{i_2}$ is not needed at step $s$.

However, as defined, U1 and U2 do not impose any conditions about the subiterations (indexed by $t$). The next lemma gives a different
characterization of the usefulness conditions in terms of $t$.

\begin{lemma}
	If $x^{i_2}$ is useful wrt to $x^{i_1}$ at some step $s$, then at some subiteration $t$ of step $s$, one of u1, u2, u3 is true
	at the start of $t$:
	\begin{itemize}
		\item[u1.] $P_{i_2,s}^{(t)} \ne x^{i_2- i_1} P_{i_1,s}^{(t)}$ 
		\item[u2.] if $P_{i_2,s}^{(t)} = x^{i_2- i_1} P_{i_1,s}^{(t)}$, then $k_{i_2}^{(t)} \ne k_{i_1}^{(t)}$
		\item[u3.] if $P_{i_2,s}^{(t)} = x^{i_2- i_1} P_{i_1,s}^{(t)}$ and $k_{i_2}^{(t)} = k_{i_1}^{(t)}$, then
		$\seq_{e,i_1,s}^{(t)}[k_{i_1}^{(t)}] \notin \mathcal{I}[k_{i_1}^{(t)}]$ and $s_{e,i_2,s}^{(t)}[k_{i_1}^{(t)}] \in \mathcal{I}[k_{i_1}^{(t)}]$
	\end{itemize}
\end{lemma}
\begin{proof}
	We prove that if u1, u2, and u3 are false for every subiteration $t$ and $s$,
	then U1 and U2 are false for $x^{i_2}$ wrt $x^{i_1}$.
	Suppose the conditions u1, u2, and u3 are all false for every subiteration $t$ at $s$. The negation of u1 forces $P_{i_2,s}^{(t)} = x^{i_2- i_1} P_{i_1,s}^{(t)}$
	at the start of $t$, which sets the hypothesis of u2 true, implying $k_{i_2}^{(t)} = k_{i_1}^{(t)}$. 
	Finally, since the hypothesis of u3 holds, we must have 
	$s_{e,i_1,s}^{(t)}[k_{i_1}^{(t)}] \in \mathcal{I}[k_{i_1}^{(t)}]$ or $s_{e,i_2,s}^{(t)}[k_{i_1}^{(t)}] \notin \mathcal{I}[k_{i_1}^{(t)}]$.
	The two are mutually exclusive since 
	$\seq_{e,i_2,s}^{(t)} = x^{i_2-i_1} \seq_{e,i_1,s}^{(t)}$, if $s_{e,i_1,s}^{(t)}[k_{i_1}^{(t)}] \in \mathcal{I}[k_{i_1}^{(t)}]$, then
	$s_{e,i_2,s}^{(t)}[k_{i_1}^{(t)}] \in \mathcal{I}[k_{i_1}^{(t)}]$. 
	When $s_{e,i_1,s}^{(t)}[k_{i_1}^{(t)}] \in \mathcal{I}[k_{i_1}^{(t)}]$, we can update
  \vspace{-0.1cm}
	\begin{align*}
	P_{i_1,s}^{(t+1)} &= P_{i_1,s}^{(t)} - \sum c_j \mathcal{I}[k_{i_1}^{(t)},j]\\
	P_{i_2,s}^{(t+1)} &= x^{i_2-i_1} P_{i_1,s}^{(t)} - x^{i_2-i_1} \sum c_j \mathcal{P}[k_{i_1}^{(t)},j] = x^{i_2-i_1} P_{i_1,s}^{(t+1)},
	\end{align*}
	which was already implied by the assumption that u1 is false for all $t$.
	On the other hand,
	when $s_{e,i_2,s}^{(t)}[k_{i_1}^{(t)}] \notin \mathcal{I}[k_{i_1}^{(t)}]$, 
	we also have $\seq_{e,i_1,s}^{(t)}[k_{i_1}^{(t)}] \notin \mathcal{I}[k_{i_1}^{(t)}]$, so the subiterations
	terminate and we must have  $P_{i_2,s} = x^{i_2- i_1} P_{i_1,s}$ with $k_{i_2} = k_{i_1}$. This implies U1 and U2 also do not hold for step $s$.
\end{proof}

While the converse is not true, we say a monomial $x^{i_2}$ is \emph{potentially useful} wrt $x^{i_1}$ when at some step $s$ and subiteration $t$,
at least one of the conditions u1, u2, and u3 holds. Rather than iterating through $i = 0,\dots,d-1$, we keep a list of potentially useful monomials $\mathcal{U}$
and iterate through $i \in \mathcal{U}$, with $\mathcal{U} = [0]$ initially. At each subiteration, we check to see if there exists $i' > i, i' \notin \mathcal{U}$
such that $x^{i'}$ satisfies one of u2 or u3, and add the smallest such $i'$ to $\mathcal{U}$. Note that we need not check u1 since if u1 holds, then either
u2 or u3 must have been true at some previous subiteration, thus $i'$ is already included in $\mathcal{U}$. Condition u2 can be checked in $\bigO{n}$ by
checking the valuations of all entries in $\seq_{e,i,s}[k]$ at lines \ref{line:check_zero1} and \ref{line:check_zero2}. Condition u3 can be checked
in $\bigO{\log d}$ membership computations via a binary search to find the minimal $i'$ such that $x^{i' -i } \seq_{e,i,s}[k] \in \mathcal{I}[k]$ when 
$ \seq_{e,i,s}[k] \notin \mathcal{I}[k]$ on line \ref{line:check_member}. Thus, the complexity for the subiterations do not change in terms of $\softO{\cdot}$. 
Defining $d^* = |\mathcal{U}| \le d$, this brings the total cost to $\softO{e d^* (n^2 e d + n^\omega d)}$. While we do not know how far $d^*$ 
is from the number 
 $d_{\mathrm{opt}}$
of polynomials in the minimal lexicographic Gr\"obner basis of $\bar\varphi(\ann{\seq})$,
we have observed empirically that $d^*$ is often equal or close to $d_{\mathrm{opt}}$
(see Section \ref{sec:experiment}).

\section{Via univariate approximant bases}
\label{sec:pmbasis}

\subsection{Approximants of a wide Hankel matrix}
\label{sec:pmbasis:basic}

Extending the classical theory of linearly recurrent sequences over the field
\(\field\), another approach is to consider the left kernel of the block-Hankel
matrix
\vspace{-0.1cm}
\[
  \hk{\seq,\tdim} =
  \begin{bmatrix}
    \seqe{0} & \seqe{1} & \cdots & \seqe{\tdim-1} \\
    \seqe{1} & \seqe{2} & \iddots & \seqe{\tdim} \\
    \vdots & \iddots & \iddots & \vdots \\
    \seqe{\tdim} & \seqe{\tdim+1} & \cdots & \seqe{2\tdim-1}
  \end{bmatrix}
  \in \alg^{(\tdim+1) \times (\tdim \sdim)}.
\]
Indeed, if \(\tdim\) is large enough, vectors in this kernel represent
polynomials which cancel $\seq$, and which even generate all of \(\ann{\seq}\).

\begin{lemma}
  \label{lem:hankel_kernel}
  Let $\seq \in \seqSet$ be linearly recurrent of order $\order$, and define
  \[
    \kerhk{\seq,\tdim} =
    \{p = p_0 + \cdots + p_\tdim y^\tdim \in \yAlg \mid [p_0 \; \cdots \; p_{\tdim}]\hk{\seq,\tdim} = 0\}
  \]
  for \(\tdim \in \NN\). Assume \(\tdim \ge \order\). Then \(\kerhk{\seq,\tdim}
  = \ann{\seq} \cap \yAlg_{\le\tdim}\), and in particular
  \(\kerhk{\seq,\tdim}\) is a generating set of \(\ann{\seq}\).
\end{lemma}
\begin{proof}
  Let \(p = p_0 + \cdots + p_\tdim y^\tdim \in \yAlg\) and \(\gamma = \deg(p)
  \le \tdim\). Then \(p \in \kerhk{\seq,\tdim}\) if and only if \([p_0 \;
  \cdots \; p_{\tdim}]\hk{\seq,\tdim} = 0\), and by definition of canceling
  partial sequences this exactly means that \(p\) cancels
  \(\tseq{\tdim+\gamma}\). Now, \(\deg(p) = \gamma \le \tdim+\gamma-\order\)
  holds under the assumption \(\tdim\ge\order\), hence \(p\) cancels
  \(\tseq{\tdim+\gamma}\) if and only if \(p \in \ann{\seq}\) by
  \cref{lem:characterization_partial}. It follows that \(\kerhk{\seq,\tdim}\)
  generates \(\ann{\seq}\), since there exists a generating set of
  \(\ann{\seq}\) whose polynomials all have degree at most \(\order\).
\end{proof}

Computing the left kernel of \(\hk{\seq,\tdim}\) can be done via univariate
approximation. Indeed, calling \(F \in \xmats{(\tdim+1)}{(\tdim\sdim)}\) the
natural lifting of \(\hk{\seq,\tdim}\), an approximant basis of \(F\) at order
\(d\) gives a generating set of that left kernel. As recalled in
\cref{sec:preliminaries:tools}, using \Call{PM-Basis}{}, a basis of
\(\modApp{d}{F}\) in shifted reduced or Popov form can be computed in
\(\softO{\tdim^{\expmm-1} (\tdim+\tdim\sdim) \prc} = \softO{\tdim^{\expmm}
\sdim \prc}\) operations in \(\field\).

\subsection{Speed-up by compression using structure}
\label{sec:pmbasis:compress}

Now we show that, when \(\sdim\) is large, one can speed up the above approach
by a randomized ``compression'' of the matrix \(\hk{\seq,\tdim}\). Precisely,
taking a random constant matrix \(C \in \mats{(\tdim\sdim)}{(\tdim+1)}\) and
performing the right-multiplication \(FC\), one obtains a square
\((\tdim+1)\times(\tdim+1)\) matrix such that \(\modApp{d}{F} =
\modApp{d}{FC}\) holds with good probability. The cost of the approximant basis
computation is thus reduced to \(\softO{\tdim^\expmm \prc}\) operations in
\(\field\), and the right-multiplication can be done efficiently by leveraging
the block-Hankel structure of \(F\).

\begin{theorem}
  \label{thm:pmbasis-compressed}
  \cref{algo:pmbasis-compressed} takes as input an integer \(\prc\in\ZZp\),
  vectors \(F_0,\ldots,F_{\mu+\tdim-2} \in \xmats{1}{\sdim}\) of degree less
  than \(\prc\), and a shift \(w\in\ZZp^\mu\), and uses
  \(\softO{\mu\tdim\sdim\prc + \mu^\expmm d}\) operations in \(\field\) to
  compute a \(w\)-Popov matrix \(P\in\xmats{\mu}{\mu}\) of degree at most
  \(d\).  It chooses at most \(\mu\tdim\sdim\) elements independently and
  uniformly at random from a subset of \(\field\) of cardinality \(\kappa\),
  and \(P\) is the \(w\)-Popov basis of \(\modApp{\prc}{F}\) with probability
  at least \(1-\frac{\mu}{\kappa}\), where $F$ is the block-Hankel matrix
  \begin{equation}
    \label{eqn:hankel_approx_input}
    F =
    \begin{bmatrix}
      F_{0} & F_{1} & \cdots & F_{\tdim-1} \\
      F_{1} & F_{2} & \iddots & F_{\tdim} \\
      \vdots & \iddots & \iddots & \vdots \\
      F_{\mu-1} & F_{\mu} & \cdots & F_{\mu+\tdim-2}
    \end{bmatrix}
    \in \xmats{\mu}{(\tdim \sdim)}.
  \end{equation}
\end{theorem}

When applied to the computation of \(\ann{\seq}\) with \(\mu=\tdim+1\), the
cost becomes \(\softO{\tdim^2\sdim\prc + \tdim^\expmm\prc}\). Below we focus on
the case of interest \(\mu\le\tdim\sdim\), since when \(\tdim\sdim \in
\bigO{\mu}\) this \(w\)-Popov approximant basis is computed deterministically
by \Call{PM-Basis}{} at a cost of \(\softO{\mu^\expmm d}\) operations in
\(\field\). Our approach is based on the following two lemmas.

\begin{lemma}
  \label{lem:pmbasis-wide}
  Let \(F \in \xmats{\mu}{\nu}\) and \(\prc \in \ZZp\). Let \(C \in
  \xmats{\nu}{r}\) and \(K \in \xmats{\nu}{(\nu-r)}\), for some \(r \in
  \{0,\ldots,\nu\}\), such that \(FK = 0\) and \([C(0) \;\; K(0)] \in
  \field^{\nu \times \nu}\) is invertible. Then, \(r\ge \rho\) where \(\rho\)
  is the rank of \(F\), and \(\modApp{\prc}{F} = \modApp{\prc}{FC}\).
\end{lemma}
\begin{proof}
  Let \(N = [C \;\; K] \in \xmats{\nu}{\nu}\). The assumption that \(N(0)\) is
  invertible ensures that \(N\) is nonsingular (since \(\det(N)(0) = \det(N(0))
  \neq 0\)), and therefore \(K\) has full rank \(\nu-r\). The assumption that
  the columns of \(K\) are in the right kernel of \(F\), which has rank
  \(\nu-\rho\), implies that \(\nu-r \le \nu-\rho\) and therefore \(r\ge \rho\).

  The inclusion \(\modApp{\prc}{F} \subset \modApp{\prc}{FC}\) is obvious. For
  the other inclusion, let \(p \in \modApp{\prc}{FC}\), i.e.~there exists \(q
  \in \xmats{1}{r}\) such that \(pFC = x^\prc q\). It follows that \(pFN =
  x^\prc [q \;\; 0]\), and thus
  \[
    pF = x^\prc [q \;\; 0] N^{-1}
    = \frac{x^\prc [q \;\; 0] \mathrm{Adj}(N)}{\det(N)}
  \]
  where \(\mathrm{Adj}(N) \in \xmats{\nu}{\nu}\) is the adjugate of \(N\). Our
  assumption \(\det(N)(0) \neq 0\) means that \(x^\prc\) and \(\det(N)\) are
  coprime, hence \(\det(N)\) divides \([q \;\; 0] \mathrm{Adj}(N)\), and \(pF =
  0 \bmod x^\prc\) follows.
\end{proof}

\begin{lemma}
  \label{lem:pmbasis-proba}
  Let \(F \in \xmats{\mu}{\nu}\) with rank \(\rho\) and \(\mu\le\nu\), and let
  \(r\in\{\rho,\ldots,\mu\}\). Let \(\mathcal{R}\) be a finite subset of
  \(\field\) of cardinality \(\kappa\in\ZZp\), and let \(C \in \mats{\nu}{r}\)
  with entries chosen independently and uniformly at random from
  \(\mathcal{R}\). Then, the probability that there exists \(K \in
  \xmats{\nu}{(\nu-r)}\) such that \([C \;\; K(0)]\) is invertible and \(FK =
  0\) is at least  \(1-\frac{r}{\kappa}\); furthermore if \(\field\) is finite and
  \(\mathcal{R}=\field\), this probability is at least
  \(\prod_{i=1}^{r}(1-\kappa^{-i})\).
\end{lemma}
\begin{proof}
  Consider a right kernel basis \(B \in \xmats{\nu}{(\nu-\rho)}\) for \(F\).
  Then \(B\) has unimodular row bases \cite[Lem.\,3.1]{ZhoLab13}, implying that
  there exists \(V \in \xmats{(\nu-\rho)}{\nu}\) such that \(VB = I_{\nu-\rho}\).
  In particular \(V(0)B(0) = I_{\nu-\rho}\) and therefore \(B(0)\) has full
  rank \(\nu-\rho\). Define \(K \in \xmats{\nu}{(\nu-r)}\) as the matrix formed
  by the first \(\nu-r\) columns of \(B\) (recall \(\nu-r \le \nu-\rho\) by
  assumption). Then \(FK=0\). Furthermore \(K(0)\) has rank \(\nu-r\), hence
  the DeMillo-Lipton-Schwartz-Zippel lemma implies that \([C \;\; K(0)] \in
  \mats{\nu}{\nu}\) is singular with probability at most \(r/\kappa\)
  \cite{DeMilloLipton78,Schwartz80,Zippel79}. If \(\field\) is finite and
  \(\mathcal{R}=\field\) then \([C \;\; K(0)]\) is invertible with probability
  exactly \(\prod_{i=1}^{r}(1-\kappa^{-i})\).
\end{proof}

These lemmas lead to \cref{algo:pmbasis-compressed} and
\cref{thm:pmbasis-compressed}; indeed computing \(FC\) has quasi-linear cost
\(\softO{\mu\tdim\sdim\prc}\) thanks to the block-Hankel structure of \(F\),
and then the call \(\Call{PM-Basis}{d,FC,w}\) costs \(\softO{\mu^\expmm\prc}\)
operations as recalled in \cref{sec:preliminaries:tools}.

\begin{algorithm}[htb]
  \caption{\textsc{Hankel-PM-Basis}\((d,F,w)\)}
  \label{algo:pmbasis-compressed}
  \begin{algorithmic}[1]
    \Require{integers \(\prc,\mu,\tdim,\sdim \in\ZZp\), vectors \(F_0,\ldots,F_{\mu+\tdim-2} \in \xmats{1}{n}\) of degree less than \(\prc\), a shift \(w \in \ZZp^\mu\)}
    \Ensure{a \(w\)-Popov matrix \(P \in \xmats{\mu}{\mu}\) of degree at most \(\prc\)}
    \State \(F \in \xmats{\mu}{(\tdim \sdim)} \gets \) form the block-Hankel matrix as in \cref{eqn:hankel_approx_input}
    \InlineIf{\(\mu \ge \tdim\sdim\)}{\Return \(\Call{PM-Basis}{d,F,w}\)}
    \State Choose \(r \in \{\rho,\ldots,\mu\}\) where \(\rho\) is the rank of
    \(F\) (by default, choose \(r=\mu\) if no information is known on \(\rho\))
    \State Fill a matrix \(C \in \mats{(\tdim\sdim)}{r}\) with entries 
    chosen uniformly and independently at random from a subset of \(\field\)
    of cardinality \(\kappa\)
    \State Compute \(FC \in \xmats{\mu}{r}\) (exploiting the Hankel structure of \(F\))
    \State \Return \(\Call{PM-Basis}{d,FC,w}\)
  \end{algorithmic}
\end{algorithm}

Note that \(1 - r/\kappa \ge 3/4\) as soon as \(\kappa\ge4\mu\) (which implies
\(\kappa\ge4r\)); furthermore \(\prod_{i=1}^{r}(1-\kappa^{-i})\ge3/4\) already
for \(\kappa=7\). The randomization is of the Monte Carlo type, since the
algorithm may return \(P\) which is not a basis of \(\modApp{\prc}{F}\). Still,
since the expected \(w\)-Popov basis \(P\) of \(\modApp{\prc}{F}\) is unique,
one can easily increase the probability of success by repeating the randomized
computation and following a majority rule. Another approach is to rely on the
non-interactive, Monte Carlo certification protocol of \cite{GiorgiNeiger2018},
which has lower cost than \cref{algo:pmbasis-compressed} but requires a larger
field \(\field\); this first asks to compute the coefficient of degree \(d\) of
\(PF\), which here can be done via bivariate polynomial multiplication in time
\(\softO{\mu\tdim\sdim\prc}\) thanks to the structure of \(F\). For a given
output \(P\), this certification can be repeated for better confidence in \(P\)
(in which case the coefficient of degree \(d\) of \(PF\) needs only be computed
once).

\section{Via bivariate Pad\'e approximation}
\label{sec:bivariate_pade}

Now, we propose another approach which directly uses the interpretation of
canceling polynomials as denominators of the generating series of the sequence
(see \cref{lem:series_characterization}). The next lemma describes more
precisely the link between the annihilator and these denominators when we have
access to a partial sequence, that is, denominators of the generating series
truncated at some order. One can also view this lemma as a description of the
kernel of the univariate Hankel matrix \(\hk{\seq,\tdim}\) via bivariate Pad\'e
approximation.

\begin{lemma}
  \label{lem:series_characterization_truncated}
  Let $\seq \in \seqSet$ be linearly recurrent of order \(\order\), and for
  \(\tdim\in\NN\) define \(G = \sum_{j < 2\tdim} \seqe{j} y^{2\tdim-1-j} \in
  \yAlg^\sdim\) and
  \[
    \mathcal{P}_{\seq,\tdim} =
    \{ p \in \yAlg_{\le \tdim} \mid pG = q \bmod y^{2\tdim}
      \text{ for some } q \in \yAlg^{\sdim}_{<\tdim} \}.
  \] 
  Assume \(\tdim \ge \order\). Then \(\mathcal{P}_{\seq,\tdim} = \ann{\seq}
  \cap \yAlg_{\le\tdim}\), and in particular \(\mathcal{P}_{\seq,\tdim}\) is a
  generating set of \(\ann{\seq}\); furthermore for any \(p \in
  \mathcal{P}_{\seq,\tdim}\) the corresponding \(q \in \yAlg^{\sdim}_{<\tdim}\)
  satisfies \(\deg(q)<\deg(p)\).
\end{lemma}
\begin{proof}
  Let \(p=p_0+\cdots+p_\gamma y^\gamma \in \yAlg_{\le \tdim}\) where \(\gamma =
  \deg(p)\). Then \(p \in \mathcal{P}_{\seq,\tdim}\) if and only if the
  coefficient of \(pG\) of degree \(2\tdim-1-k\) is zero for \(0 \le k <
  \tdim\). Since \(\gamma \le \tdim \le 2\tdim-1-k\), this coefficient is
  \[
    \mathrm{Coeff}(pG,2\tdim-1-k)
    = \sum_{i=0}^{\gamma} p_i S_{2e-1-(2e-1-k-i)}
    = \sum_{i=0}^{\gamma} p_i S_{k+i}
    = 0.
  \]
  Thus we have proved \(\mathcal{P}_{\seq,\tdim} = \kerhk{\seq,\tdim}\), and
  \cref{lem:hankel_kernel} shows the claims in this lemma except the last one.
  Let \(p \in \mathcal{P}_{\seq,\tdim}\) and define \(q\) as the polynomial in
  \(\yAlg^{\sdim}_{<\tdim}\) such that \(pG = q \bmod y^{2\tdim}\). Since
  \(p\in\ann{\seq}\), \cref{lem:series_characterization} shows that \(p\gens\) is a
  polynomial. On the other hand the definitions of \(G\) and \(\gens\) yield \(pG =
  y^{2\tdim}p\gens - p\sum_{j\ge 2\tdim} \seqe{j} y^{2\tdim-1-j}\). Hence \(-
  p\sum_{j\ge 2\tdim} \seqe{j} y^{2\tdim-1-j}\) is a polynomial, and since it
  has degree less than \(\gamma\), and thus in particular less than \(2\tdim\),
  it is equal to \(q\).
\end{proof}

From \(G\), define \(F \in \abmats{1}{\sdim}\) of bi-degree less than
\((\prc,2\tdim)\) via the morphism \(\varphi\) from
\cref{sec:preliminaries:bivariate}. Equip \(\abRing\) with the lexicographic
order \(\ordLex\), and let \(\ord\) be the corresponding term over position
order on \(\abRing^{\sdim+1}\).
Then a minimal \(\ord\)-Gr\"obner basis of the submodule of simultaneous Pad\'e
approximants
\[
  \{ (p,q) \in \abRing\times\abmats{1}{\sdim} \mid pF = q \bmod (x^\prc,y^{2\tdim}) \}
\]
is computed in \(\softO{(\sdim^\expmm\min(\prc,\tdim)^\expmm +
\sdim^3\min(\prc,\tdim)^2) \prc \tdim}\) operations, using the algorithm of
\cite{NaldiNeiger2020} (see also \cref{sec:preliminaries:tools}) with input
matrix of size \((\sdim+1) \times \sdim\) formed by stacking the identity
\(I_\sdim\) below \(F\). \cref{lem:series_characterization_truncated} shows
that from this \(\ord\)-Gr\"obner basis one can find a minimal
\(\ordLex\)-Gr\"obner basis of \(\bar\varphi(\ann{\seq})\) by selecting \(p\)
for each \((p,q)\) in the basis such that
\(\deg_{\beta}(q) < \deg_{\beta}(p)\).

While the \Call{PM-Basis}{} approach had cost quasi-linear in \(\prc\) and
\(\sdim\), the method here is most efficient in an opposite parameter range:
for \(\sdim\in\bigO{1}\) and \(\prc\le\tdim\) the above cost bound becomes
\(\softO{\prc^{\expmm+1} e}\).

\section{Experimental Results }
\label{sec:experiment}

In this section, we compare timings for the algorithms in
\cref{sec:kurakin,sec:modified_kurakin,sec:pmbasis}, implemented in C++ using
the libraries NTL \cite{NTL} and PML \cite{PML} which provide high-performance
support for univariate polynomials and polynomial matrices.  We leave the
implementation of the bivariate algorithm of \cref{sec:bivariate_pade} as
future work.
To control the cardinality and shape of the Gr\"obner basis, we
use Lazard's structural theorem (see \cref{sec:preliminaries:bivariate}).
The shape of the monomial staircase is randomized with maximal $\beta$-degree
$\order$ and $\alpha^d$ included in the basis.  After generating a random
Gr\"obner basis $\mathcal{G}$ of target degree and size, we use it to generate $n$
sequences (with $e = 2\delta$ terms), using random initial conditions. Finally, we
compute the annihilator of the sequence,
which may not necessarily recover $\mathcal{G}$ itself (see \cref{sec:minpoly}).
Runtimes are showed below.

\noindent\resizebox{\columnwidth}{!}{
  \begin{tabular}{|c|c|c||c|c||c|c|c||c|c|}
  \hline
  $n$ & $d$ & $\order$ & $d_{\mathrm{opt}}$ & $D/d\order$ & K & LK & $d^*$ & PM-B & HPM\\
  \hline
  1  & 64  & 256 & 1   &  1         & 62.8 & 0.93 & 1  & 1.06 & NA \\
  1  & 64  & 256 & 49  &  0.62      & 38.0 & 1.65 & 53 & 2.10 & NA \\ 
  1  & 128 & 512 & 16  &  0.92      & >100 & 12   & 17 & 20.5 & NA \\
  1  & 128 & 32  & 12  &  0.91      & 7.85 & 0.078 & 12 & 0.029 & NA \\
  1  & 256 & 32  & 14  &  0.94      & 27.3 & 0.12 & 14 & 0.08 & NA \\
  1  & 256 & 128 & 27  &  0.92      & >100 & 1.28 & 27 & 1.60 & NA \\
  1  & 512 & 256 & 29  &  0.96      & >100 & 8.65 & 29 & 27.8 & NA\\ 
  2  & 17  & 256 & 2   &  0.5       & 14.1 & 0.91 & 2 & 0.33 & 0.29 \\ 
  3  & 12  & 512 & 4   &  0.4       & 6.93 & 1.40 & 4 & 2.47 & 1.86 \\
  8  & 16  & 256 & 1   &  1         & 54.1 & 3.16 & 1  & 0.56 & 0.25 \\
  32 & 16  & 256 & 1   &  1         & >100 & 39.8 & 1 & 2.79 & 0.35\\
  64 & 16  & 128 & 1   &  1         & >100 & >100 & 1 & 1.02 & 0.13\\
  \hline
  \end{tabular}
}
\textbf{Table:} \emph{Runtimes, in seconds, of algorithms Kurakin, Lazy
Kurakin, direct \textsc{PM-Basis}, and \textsc{Hankel-PM-Basis}, observed on
AMD Ryzen 5 3600X 6-Core CPU with 16 GB RAM.  The base field is \(\mathbb{K} =
\mathbb{F}_{9001}\).}

\smallskip

As we claim in \cref{sec:modified_kurakin}, $d^*$ is often close or equal to $d_{\mathrm{opt}}$. More interestingly,
Lazy Kurakin outperforms Kurakin more than $d / d^*$ would suggest. For example, for $\delta = 256, d = 64, d_{\mathrm{opt}} = 49$,
then $d/d^* \approx 1.2$ but Kurakin is 23 times slower than Lazy Kurakin. 
This is because the cost bound $\softO{ed^*(n^2 e d + n^\omega d)}$
for Lazy Kurakin assumes that $d^*$ polynomials are tracked from the beginning of the algorithms. However, due to its lazy nature, polynomials are
often added later in the algorithm and the bound of $ed^*$ subiterations may significantly overestimate the true number of subiterations.

When $\order, d, n$ are fixed, Kurakin's algorithm performs worse for $d_{\mathrm{opt}} = 1$ than $d_{\mathrm{opt}} > 1$, although
this is a favourable case for Lazy Kurakin. In this case, Kurakin's algorithm computes $P_i = x^i P_0$ so there cannot be any early termination.
Additionally, the size of the staircase is maximal ($D = ed$), so this is also the worst case for algorithms whose complexity depends directly on $D$. 
Lazy Kurakin's algorithm somewhat remedies this by using the extra structure of $\alg$ and adding monomials
in a lazy fashion.
(When it is known that $\ann{\seq} = \langle P \rangle$, it is possible to design an algorithm that is quasilinear in $e$ via structured system solving,
see \cref{sec:app:det}).

For scalar sequences over $\alg$, i.e.~$n=1$, Lazy Kurakin's algorithm seems to be the best choice when $\delta$ is large
compared to $d$, whereas PM-Basis seems to be the best choice in the converse. When $e = 2\delta = d$, Lazy Kurakin outperforms
PM-Basis, given that $d^*$ is small. This is predicted by the theoretical complexities, as the former has complexity
$\softO{e^3 d^*}$, while the latter has complexity $\softO{e^{\omega+1}}$.

For $n > 1$, \textsc{PM-Basis} and \textsc{Hankel-PM-Basis} clearly outperform Kurakin and Lazy Kurakin. 
This is as predicted since the complexity of the former
depends linearly on $n$, while the latter has a factor $n^\omega$. 
The theoretical improvement of \textsc{Hankel-PM-Basis} over \textsc{PM-Basis} is
observed empirically, especially for the two cases of $n = 32, 64$.

\section{Applications to sparse matrices} \label{sec:applications}

In this section, we outline two applications to sparse matrices $A \in
\alg^{n\times n}$: first, the computation of minimal polynomials of \(A\),
which are polynomials of minimal degree that cancel the matrix sequence $\seq_A
= (A^0, A^1, A^2,\ldots)$; second, the computation of the determinant of \(A\).
In what follows, we assume $A$ has sparsity $\bigO{n}$, i.e.~it has
\(\bigO{n}\) nonzero entries, and that the representation of \(A\) allows us to
compute matrix-vector products at cost $\softO{nd}$. Our approach is based on
Wiedemann's \cite{wiedemann}, designed for matrices over fields.

\subsection{Minimal polynomials of sparse matrices}
\label{sec:minpoly}

Given a matrix $A$, the well-known Cayley-Hamilton theorem states that $A$ cancels its own characteristic polynomial.
This implies that the sequence of successive powers of $A$ is linearly recurrent, and a polynomial of minimal degree that cancels this sequence
is said to be a minimal polynomial of $A$.
A different view one can take is that such canceling polynomials must
cancel the $n^2$ linearly recurrent sequences $((A^i)_{j_1,j_2})_{i\ge 0}$ simultaneously for $1 \le j_1,j_2 \le n$.
Then, as usual, we want to compute a Gr\"obner basis of the ideal of these canceling polynomials, denoted by $\ann{A}$.

Over $\alg$, trying to deduce $\ann{A}$ from $\ann{ (u^T A^i v)_{i\ge 0}}$, for random vectors $u,v \in \alg^{n \times 1}$, presents
a problem when $\ann{A}$ does not have the \emph{Gorenstein} property \cite{macaulay,grobner35}. When $\ann{A}$ has the Gorenstein property,
it has been showed that $\ann{A}$ can be recovered, with high probability, 
by using a bidimensional sequence with random initial conditions, provided $\field$ has large characteristic \cite{berthomieu17}. When it does not have the property,
$\ann{A}$ is still recoverable with a similar approach, but using several sequences \cite{NeHaSc17}.
Over various commutative rings, the problem of computing minimal polynomials of a matrix have been studied in
\cite{brown,heuberger,rissner}. However, the algorithms given in these works do not exploit sparsity. 

Given matrix $A$ as above, we start by choosing random $u_1,v \in \alg^n$ and generating $\seq_{A,1} = (u_1^T A^i v)_{0 \le i < 2n}$. 
Next, we apply one of the algorithms in the previous sections to compute
$\ann{\seq_{A,1}}$. If $\ann{\seq_{A,1}} = \ann{A}$, which can be checked probabilistically by checking if $\ann{\seq_{A,1}}$ also cancels
some validation sequence $((u')^T A^i v)_{0 \le i < 2n}$, we terminate the process. Otherwise, we double the number of sequences by doubling the
number of random $u_i$'s and generating $\seq_{A,1},\ldots, \seq_{A,2^s}$. The cost of the process is 
$\softO{\tau n^2 d + \mathcal{L}(n, d, \tau)}$, where $\tau$ is the number of sequences used and $\mathcal{L}(n, d, \tau)$ is the cost of
finding the annihilators of a partial sequence of length $n$ in $(\xRing/\genby{x^\prc})^\tau$. Note that this process must terminate. 
The crudest bound is when $\tau > n^2$ 
since then we could simply compute $\ann{A}$ directly. Another slightly more refined bound for the number of generic linear forms needed is 
$\tau \le D$, where $D$ is the size of the staircase of $\ann{A}$ \cite[Prop.\,1]{NeHaSc17}.

\subsection{Determinant of sparse matrices}
\label{sec:app:det}

The determinant of a matrix is easily obtained from its minimal polynomial when
the latter is equal to the characteristic polynomial. Wiedemann
\cite{wiedemann} calls such matrices \emph{nonderogatory} and shows that
preconditioning any matrix $B \in \field^{n\times n}$ with a random diagonal
matrix $D$ results in a nonderogatory matrix with high probability. We will
show that the same preconditioning can be applied to matrices over $\alg$.
Here, a particular role will be played by sequences $\seq \in \seqSetExpand$
such that $\ann{\seq} = \langle P \rangle$, for some monic $P \in \yAlg$.
Indeed, the next theorem shows that it is sufficient for the constant part of
$A$ to be nonderogatory in $\field$ for $A$ to be nonderogatory in $\alg$ and
for the sequence of its powers to satisfy this property.

\begin{theorem}
  Let $A_0 \in \mats{n}{n}$ be the constant part of $A$ (i.e.~for \(x=0\)). If
  $A_0$ is nonderogatory, then $\ann{A} = \langle P \rangle$ for some monic
  $P\in \yAlg$ of degree $n$.
\end{theorem}

\begin{proof}
  Let $P\in\yAlg$ be the minimal monic polynomial of the sequence $\seq_A =
  (A^0, A^1, A^2, \ldots)$, then $\deg(P) \le n$ since $A$ is $n\times n$.
  Now, $A_0$ is nonderogatory, so any canceling polynomial must have
  degree $\ge n$; thus, $\deg(P) = n$. Furthermore, if there
  exists another polynomial $Q$ of degree $n$ and leading coefficient $x^i$
  such that $Q \ne x^i P$, then $Q - x^i P$ is a canceling polynomial of degree
  less than $n$, contradicting the previous statement.  Thus, $P, xP, \dots,
  x^{d-1} P$ are minimal in degree and, by \cref{theorem:kurakin}, $\ann{A} =
  \langle P, xP, \dots, x^{d-1}P \rangle = \langle P \rangle$.
\end{proof}

The above theorem allows us to use the same preconditioner as in
\cite{wiedemann}: a random constant diagonal matrix $D$. The preconditioning
ensures that the ideal of canceling polynomial is generated by a single monic
polynomial; thus, $\bar\varphi(\ann{AD})$ is Gorenstein and requires only a
single linear form to be recovered. Furthermore, when it is known that the
ideal is generated by a single polynomial, we can recover this polynomial in
$\softO{nd}$ by taking advantage of the fact that the constant part of the
leading $n \times n$ submatrix of $\hk{\seq,2n}$ is an invertible Hankel matrix
\cite{bostan}. Once we have $P$, we can compute $\det(A) = P(0) (\prod_i
D_{i,i})^{-1}$.
Under our sparsity assumption, the cost of this method is $\softO{n^2d}$
for computing $(u^{T} A^i v)_{i \le 2n}$, $\softO{nd}$ for computing $P$,
and $\softO{n+d}$ for recovering the determinant from $P$, leading to the
total cost of $\softO{n^2 d}$ operations in \(\field\). This is to be
compared with computing the determinant of \(A\) ``at full precision'', i.e.~by
seeing \(A\) as a matrix over \(\xRing\), and then truncating the result modulo
\(x^\prc\): this costs $\softO{n^\expmm d}$ operations in \(\field\)
\cite{LabahnNeigerZhou2017}.


\end{document}